\documentclass[11,preprint]{elsarticle}

\usepackage{framed,multirow}

\usepackage{latexsym}

\usepackage{url}

\usepackage{lipsum}
\usepackage{amsfonts,bm,amssymb,amsmath,amsthm}
\usepackage{graphicx}
\usepackage{epstopdf}
\usepackage[caption=false]{subfig}
\usepackage{pgfplots}
\usepackage{algorithmic}
\usepackage{xcolor}
\usepackage{lineno}
\usepackage{tabularx}
\usepackage{float}
\usepackage{gensymb}
\usepackage{varioref}
\usepackage{hyperref}
\usepackage{cleveref}
\usepackage{color}

\hypersetup{
colorlinks,
linkcolor={red!50!black},
citecolor={blue!50!black},
urlcolor={blue!80!black}
}

\newtheorem{remark}{Remark}
\newtheorem{definition}{Definition}

\newtheorem{thm}{Theorem}[section]
\newtheorem{lem}{Lemma}[section]

\def\e{\mbox{\boldmath $e$}}
\def\f{\mbox{\boldmath $f$}}

\def\g{\mbox{\boldmath $g$}}
\def\h{\mbox{\boldmath $h$}}
\def\m{\mbox{\boldmath $m$}}

\def\p{\mbox{\boldmath $p$}}
\def\q{\mbox{\boldmath $q$}}
\def\x{\mbox{\boldmath $x$}}

\def\0{\mbox{\boldmath $0$}}

\newcommand{\lemref}[1]{\textbf{Lemma~\ref{#1}}}
\newcommand{\thmref}[1]{\textbf{Theorem~\ref{#1}}}
\newcommand{\secref}[1]{\S\ref{#1}}

\xdefinecolor{dred}{rgb}{0.6, 0.0, 0.0}

\definecolor{newcolor}{rgb}{.8,.349,.1}
\definecolor{ggreen}{rgb}{0.0,0.5,0.0}
\colorlet{bblue}{blue!50!black}

\crefformat{equation}{(#2#1#3)}
\crefmultiformat{equation}{(#2#1#3)}{ and~(#2#1#3)}{, (#2#1#3)}{ and~(#2#1#3)}

\crefformat{figure}{Figure~#2#1#3}
\crefmultiformat{figure}{Figures~(#2.#1#3)}{ and~(#2#1#3)}{, (#2#1#3)}{ and~(#2#1#3)}
\crefformat{table}{Table~#2#1#3}

\begin{document}

\begin{frontmatter}

\title{Second-order semi-implicit projection methods for micromagnetics simulations}

\author[1]{Changjian Xie}
\ead{20184007005@stu.suda.edu.cn}

\author[2,3]{Carlos J. Garc\'{i}a-Cervera}
\ead{cgarcia@math.ucsb.edu}

\author[c]{Cheng Wang}
\ead{cwang1@umassd.edu}

\author[d]{Zhennan Zhou\corref{cor1}}
\ead{zhennan@bicmr.pku.edu.cn}

\author[1,e]{Jingrun Chen\corref{cor1}}
\cortext[cor1]{Corresponding authors.}
\ead{jingrunchen@suda.edu.cn}

\address[1]{School of Mathematical Sciences, Soochow University, Suzhou, 215006, China.}
\address[2]{Department of Mathematics, University of California, Santa Barbara, CA 93106, USA.}
\address[3]{Basque Center for Applied Mathematics, Mazarredo 14, E48009 Bilbao, Basque Country, Spain.}
\address[c]{Mathematics Department, University of Massachusetts Dartmouth, North Dartmouth, MA 02747, USA.}
\address[d]{Beijing International Center for Mathematical Research, Peking University, Beijing, China.}
\address[e]{Mathematical Center for Interdisciplinary Research, Soochow University, Suzhou, 215006, China.}

\begin{abstract}
Micromagnetics simulations require accurate approximation of the magnetization dynamics described by the Landau-Lifshitz-Gilbert equation,
which is nonlinear, nonlocal, and has a non-convex constraint, posing interesting challenges in developing numerical methods. In this paper,
we propose two second-order semi-implicit projection methods based on the second-order backward differentiation formula
and the second-order interpolation formula using the information at previous two temporal steps. Unconditional unique solvability
of both methods is proved, with their second-order accuracy verified through numerical examples in both 1D and 3D.
The efficiency of both methods is compared to that of  another two popular methods.
In addition, we test the robustness of both methods for the first benchmark problem with a ferromagnetic thin film material 
from National Institute of Standards and Technology.
\end{abstract}


\end{frontmatter}

\section{Introduction}

Electrons in a material pose local magnetic orders, but typically do not exhibit a
spontaneous macroscopic magnetic ordering unless a collective motion of these local magnetic orders is present. This results in a net magnetization even in the absence of an external magnetic field. Such a material is called a ferromagnet. It has binary stable
configurations, which makes it an ideal material for data recording and storage.
Recent advances in experiment and theory \cite{ZuticFabianDasSarma:2004} have demonstrated
effective and precise control of ferromagnetic configurations by means of external fields.

A very common phenomenological model for magnetization dynamics is  the Landau-Lifshitz-Gilbert (LLG) equation
\cite{LandauLifshitz:1935, Gilbert:1955}. This model has been successfully used to interpret various experimental observations. 
The LLG equation is technically quasilinear, nonlocal and has a non-convex constraint, which poses interesting challenges in designing efficient numerical methodologies. 
In addition, the magnetization reversal process requires numerical methods to resolve different length and temporal scales in the presence of domain walls 
and vortices, due to their important roles in the switching process \citep{Shi1999mag, Shi2000geometry}. Therefore, numerical methods
for the LLG equation with high accuracy and efficiency are highly demanding.

There has been a continuous progress of developing numerical algorithms in the past few decades; see for example \citep{Kruzik:2006,cimrak2007survey} and references therein.
The spatial derivative is typically approximated by the finite element method (FEM) \citep{Abert2014, MR3260281, PageDecoupled:2014, Aurada:2015, KimPage2015, BanasPage:2015, Praetorius:2018, alouges2006convergence, AlougesFEM:2008, Alouges2014, AlougesStoch2014} and the finite difference method \citep{weinan2001numerical, fuwa2012finite,Jeong2010613,KIM2017109,ROMEO2008464}.
As for the temporal discretization, explicit schemes \citep{alouges2006convergence,jiang2001hysteresis}, 
fully implicit schemes \citep{Prohl2001Computational, bartels2006convergence,fuwa2012finite}, and semi-implicit schemes
\citep{weinan2001numerical,wang2001gauss,Garcia2001Improved,lewis2003geometric, cimrak2005error, gao2014optimal, Boscarino2016High} have been extensively studied.
Explicit schemes suffer from severe stability constraints. Fully implicit schemes can overcome this severe stability constraints. However, a (nonsymmetric) 
nonlinear system of equations needs to be solved at each time step, which is time-consuming. 
A nonlinear multigrid method is used to handle the nonlinearity at each time step in \citep{jeong2014accurate},
and the fixed point iteration technique is used to deal with the nonlinearity in \citep{cimrak2004iterative}.
In \citep{fuwa2012finite}, the existence and uniqueness of a solution to the nonlinear system is proved under the condition that the temporal stepsize be proportional to the square of the spatial gridsize.
This, however, is contrary to the unconditional stability of implicit schemes.

Semi-implicit schemes achieve a desired balance between stability and efficiency. One of the most popular methods is the Gauss-Seidel projection method (GSPM)
developed by Wang, Garc\'ia-Cervera, and E \citep{weinan2001numerical,wang2001gauss,Garcia2001Improved}. This method is based on a combination of a Gauss-Seidel 
implementation of a fractional step implicit solver for the gyromagnetic term, and the projection method for the heat flow 
of harmonic maps to overcome the difficulties associated with the stiffness and nonlinearity.
Only several linear systems of equations need to be solved at each step, whose complexity is comparable to solving the scalar 
heat equation implicitly. It is tested that GSPM is unconditionally stable with first-order accuracy in time.
In order to get second-order accuracy in time, two nonsymmetric linear systems of equations need to be solved at each step.
Note that a projection step is needed to preserve the pointwise length constraint.

In this work, we propose two second-order semi-implicit projection methods for LLG equation
based on the second-order backward differentiation formula
and the second-order interpolation formula using the information at previous two temporal steps. The unconditional unique solvability
of both methods is proved, with their second-order accuracy verified through numerical examples in both 1D and 3D.
The efficiency of both methods is compared to that of another two popular schemes in the literature.
In addition, we test the robustness of both methods using the first benchmark problem for a ferromagnetic thin film material
developed by the micromagnetic modeling activity group from National Institute of Standards and Technology (NIST).
It is worth mentioning that a modification of the proposed method has been proved to be second-order accurate in time
under mild conditions \cite{jingrun2019analysis}.

The rest of the paper is organized as follows. In \Cref{section:method}, we first introduce the micromagnetics model based on the LLG equation. 
The second-order semi-implicit projection methods are described in \Cref{section:scheme} with their unique solvability given in \Cref{section:solvability}. 
The calculation of the demagnetization field (stray filed) is given in \Cref{section:stray}. 
Numerical results in \Cref{section:numerical1} are used to test the accuracy and the efficiency of the proposed methods in both 1D and 3D. 
Moreover, the first benchmark for a ferromagnetic thin film material developed by the micromagnetics modeling activity group from NIST is used to check the applicability of the proposed methods in \Cref{section:numerical2}.
Conclusions are drawn in \Cref{section:conclusion}.

\section{Second-order semi-implicit methods}\label{section:method}

\subsection{Landau-Lifshitz-Gilbert equation}\label{model}

The magnetization dynamics in a ferromagnetic material are described by the LLG equation~\cite{LandauLifshitz:1935,Brown1963micromagnetics}, which take the following nondimensionalized form:
\begin{align}\label{c1}
{\m}_t =-{\m}\times{\bm h}_{\text{eff}}-\alpha{\m}\times({\m}\times{\bm h}_{\text{eff}}),
\end{align}
with
\begin{equation}\label{boundary}
\frac{\partial{\m}}{\partial {\bf \nu}}\Big|_{\Gamma}=0
\end{equation}
on $\Gamma=\partial \Omega$.
Here the magnetization ${\m}\,:\,\Omega\subset\mathbb{R}^d\to S^2,d=1,2,3 $ is a three-dimensional vector field
with a pointwise constraint $|\m|=1$ and ${\bf \nu}$ is the unit outward normal vector. $\Omega$ is a bounded domain
occupied by the ferromagnetic material. The first term on the right hand side in \cref{c1} is the gyromagnetic
term and the second term is the damping term with $\alpha>0$ being the dimensionless damping coefficient.

The effective field ${\bm h}_{\text{eff}}$ consists of the exchange field, the anisotropy field, the external field $\h_e$
and the demagnetization or stray field $\h_s$. For a uniaxial material,
\begin{align}
    {\bm h}_{\text{eff}} =\epsilon\Delta\m-Q(m_2\e_2+m_3\e_3)+\h_s+\h_e.
\end{align}
Here, the dimensionless parameters are $Q=K_u/(\mu_0 M_s^2)$ and $\epsilon=C_{ex}/(\mu_0 M_s^2L^2)$ with $L$ the diameter of the ferromagnetic body, $K_u$ the anisotropy constant, $C_{ex}$ the exchange constant, $\mu_0$ the permeability of vacuum, and $M_s$ the saturation magnetization.
${\bm e}_2=(0,1,0)$, ${\bm e}_3=(0,0,1)$ and $\Delta$ denotes the standard Laplacian operator.
$\h_e$ is the applied (external) magnetic field and the detailed description of $\h_s$ will be given in \secref{section:stray}.
Typical values of the physical parameters for Permalloy are included in \Cref{tab}.
\begin{table}[htbp]
	\centering
	\caption{Typical values of the physical parameters for Permalloy, which is an alloy of Nickel (80\%) and Iron (20\%) frequently used in magnetic storage devices.}\label{tab}
	\begin{tabular}{|c|p{5cm}<{\centering}|}
		\hline
		\multicolumn{2}{|c|}{\bf Physical Parameters for Permalloy}\\
		\hline
		$K_u$ & $1.0\times10^2\;\mathrm{J/m^3}$ \\
		\hline
		$C_{ex}$& $1.3\times10^{-11}\;\mathrm{J/m}$ \\
		\hline	
		$M_s$ & $8\times 10^5\;\mathrm{A/m}$ \\
		\hline
		$\mu_0$ & $4\pi \times 10^{-7}\;\mathrm{N/A^2}$ \\
		\hline
		$\alpha$ & $0.01$ \\
		\hline
	\end{tabular}
\end{table}
For brevity, we define
\begin{align}\label{eq-4}
\f=-Q(m_2\e_2+m_3\e_3)+\h_s+\h_e.
\end{align}
and rewrite \cref{c1} as
\begin{align}\label{eq-5}
\m_t=-\m\times(\epsilon\Delta\m+\f)-\alpha\m\times\m\times(\epsilon\Delta\m+\f).
\end{align}
It is easy to check that the following equation
\begin{equation}\label{eqq-7}
(1-\alpha{\bm m}\times){\bm m}_t=-(1+\alpha^2){\bm m}\times(\epsilon\Delta\m+\f)
\end{equation}
is equivalent to \eqref{eq-5} since $|\m|=1$.

\subsection{Second-order semi-implicit projection methods}\label{section:scheme}

Denote the temporal step-size by $k$, the spatial mesh size by $h$,  the standard second-order centered difference for Laplacian operator by $\Delta_h$, and $t^n=nk$, $n\leq \left\lfloor\frac{T}{k}\right\rfloor$
with $T$ the final time. For convenience, we use the finite difference method to approximate the spatial derivatives in \cref{eq-5,eqq-7}.
For the temporal discretization, we employ the second-order backward differentiation formulas (BDFs) to approximate the temporal derivative
\begin{equation}
\frac{\partial}{\partial t}\m_h^{n+2}\approx\frac{\frac{3}{2}\m_h^{n+2}-2\m_h^{n+1}+\frac{1}{2}\m_h^n}{k}.
\end{equation}
Such a discretization results the following fully implicit scheme for \cref{eq-5}:
\begin{align}\label{2point5}
\frac{\frac{3}{2}\m_h^{n+2}-2\m_h^{n+1}+\frac{1}{2}\m_h^n}{k}&=-\m_h^{n+2}\times \big(\epsilon \Delta_h\m_h^{n+2} +\f_h^{n+2} \big)\\
&-\alpha \m_h^{n+2} \times\big(\m_h^{n+2}\times(\epsilon \Delta_h\m_h^{n+2}+\f_h^{n+2}) \big)\nonumber.
\end{align}
As expected, at each time step, a nonlinear system of equations needs to be solved in \cref{2point5}. Moreover,
the nonsymmetric structure of the system introduces additional difficulties.

To overcome this severe difficulty, we approximate the nonlinear prefactors in front of the discrete Laplacian term
using the information from previous time steps (one-sided interpolation) with its accuracy the same as the corresponding
BDF scheme. For \eqref{2point5}, we have
\begin{align}
\frac{\frac{3}{2}{\m}_h^{n+2}-2{\m}_h^{n+1}+\frac{1}{2}{\m}_h^n}{k}
&= -\hat{\m}_h^{n+2}\times\big(\epsilon \Delta_h\m_h^{n+2} +\hat{\f}_h^{n+2} \big) \nonumber\\
&  -\alpha\hat{\m}_h^{n+2}\times\left(\hat{\m}_h^{n+2}\times(\epsilon \Delta_h\m_h^{n+2} +\hat{\f}_h^{n+2} ) \right),
\end{align}
where
\begin{align}
    \hat{\m}_h^{n+2} &=2{\m}_h^{n+1}-{\m}_h^n, \label{m_hat}\\
    \hat{\f}_h^{n+2} &=2{\f}_h^{n+1}-{\f}_h^n,
\end{align}
and $\f_h^{n}=-Q(m_2^n\e_2+m_3^n\e_3)+\h_s^n+\h_e^n$. However, such a scheme cannot preserve the magnitude of magnetization,
we therefore add a projection step and obtain the following scheme for \cref{eq-5}
\begin{equation*}\label{ccc72}
\textrm{\bf Scheme A }	\left\{ 
	    \begin{aligned}
	    &\frac{\frac32 {\m}_h^{n+2,*} - 2 {\m}_h^{n+1} + \frac12 {\m}_h^n}{k}
=  - \hat{\m}_h^{n+2} \times\big(\epsilon \Delta_h\m_h^{n+2,*} +\hat{\f}_h^{n+2} \big) \\
	    &\quad - \alpha \hat{\m}_h^{n+2} \times \left(\hat{\m}_h^{n+2}\times(\epsilon \Delta_h\m_h^{n+2,*} +\hat{\f}_h^{n+2} ) \right), \\
	   & \qquad\qquad\qquad\qquad\quad \m_h^{n+2} = \frac{{\m}_h^{n+2,*}}{ |{\m}_h^{n+2,*}| },
	    \end{aligned}
	\right.
\end{equation*} 
where ${\m}_h^{n+2,*}$ is the intermediate magnetization.
For \cref{eqq-7}, using the same idea, we have
\begin{equation*}\label{ccc73}
\textrm{\bf Scheme B }	\left\{
\begin{aligned}
(1-\alpha \hat{\m}_h^{n+2}\times)& \frac{\frac32 {\m}_h^{n+2,*} - 2 {\m}_h^{n+1} + \frac12 {\m}_h^n}{k}\\
&  =  - (1+\alpha^2)\hat{\m}_h^{n+2} \times\big(\epsilon \Delta_h\m_h^{n+2,*} +\hat{\f}_h^{n+2} \big), \\
 \m_h^{n+2} & = \frac{{\m}_h^{n+2,*}}{ |{\m}_h^{n+2,*}| }.
\end{aligned}
\right.
\end{equation*}

\begin{remark}
	Using the same idea, we can construct high-order semi-implicit projection schemes for \cref{eq-5} and
	\cref{eqq-7} using high-order BDFs and high-order one-sided interpolations. However, if the order
	is greater than $2$, then such a scheme cannot be A-stable. Numerical tests for LLG equation also
	indicates the conditional stability of higher order semi-implicit projection schemes. First-order
	semi-implicit projection schemes are A-stable, but they do not have obvious advantages over the
	Gauss-Seidel projection method \cite{wang2001gauss}.
\end{remark}

\begin{remark}
To kick start \textbf{Scheme A} and \textbf{Scheme B} in implementation, we use the first-order semi-implicit projection scheme with 
the first-order BDF and the first-order one-sided interpolation for the first temporal step, and thus the whole method is still second-order accurate.
\end{remark}

\subsection{Unconditionally unique solvability}\label{section:solvability}
For simplicity of illustration, we assume that the spatial mesh size $h_x = h_y = h_z =h$. An extension to the general case is straightforward.

We firstly introduce the discrete inner product.
\begin{definition}[Discrete inner product]
	For grid functions $\f_h$ and $\g_h$ over the uniform numerical grid, we define
	\begin{align*}
	\langle {\bm f}_h,{\bm g}_h \rangle = h^d\sum_{\mathcal{I}\in \Lambda_d} \f_{\mathcal{I}}\cdot\g_{\mathcal{I}},
	\end{align*}
	where $\Lambda_d$ is the index set and $\mathcal{I}$ is the index which closely depends on the dimension $d$.
\end{definition}

\begin{definition}
	For the grid function $\f_h$, we define the average of summation as
	\begin{align*}
	\overline{\f_h}=h^d\sum_{\mathcal{I}\in \Lambda_d}\f_{\mathcal{I}}.
	\end{align*}
\end{definition}

\begin{definition}
	For the grid function $\f_h$ with $\overline{\f_h}=0$, we define the discrete $H_h^{-1}$-norm as
	\begin{equation*}
	\|\f_h\|_{-1}^2=\langle(-\Delta_h)^{-1}\f_h,\f_h\rangle.
	\end{equation*}
\end{definition}

For ease of notation, we drop the temporal indices and rewrite \textrm{\bf Scheme A } and \textrm{\bf Scheme B }
in a more compact form
\begin{align}
\big(\frac{3}{2}I+k\epsilon\hat{\m}_h\times\Delta_h+\alpha k\epsilon\hat{\m}_h\times(\hat{\m}_h\times\Delta_h)\big){\m}_h=\p_h,\label{eqqq-19} \\
\big(\frac{3}{2}I-\frac{3}{2}\alpha\hat{\m}_h\times+k\epsilon(1+\alpha^2)\hat{\m}_h\times\Delta_h\big)\m_h=\tilde{\p}_h, \label{eqqq-20}
\end{align}
where $\p_h$, $\tilde{\p}_h$, and $\hat{\m}_h$ are given.

\begin{thm}[Solvability for \textrm{\bf Scheme A}]\label{thm2}
	Given $\p_h$ and $\hat{\m}_h$, the numerical scheme \cref{eqqq-19} admits a unique solution.
\end{thm}

For the unique solvability analysis for \cref{eqqq-19}, we denote $\q_h = - \Delta_h \m_h$. Note that $\overline{\q_h} =0$, due to the Neumann boundary condition for ${\m}_h$. Meanwhile, we observe that ${\m}_h \ne (-\Delta_h)^{-1} \q_h$ in general, since $\overline{{\m_h}} \ne 0$.
Instead, we rewrite \cref{eqqq-19} into
\begin{align*}
\m_h = \;& \frac23 \Bigl( \p_h + k \epsilon \hat{\m}_h \times \q_h
+ \alpha k \epsilon \hat{\m}_h \times ( \hat{\m}_h \times \q_h)  \Bigr)
\end{align*}
and take the average on both sides. Therefore, ${\m}_h$ can be represented as follows:
\begin{equation*}
{\m}_h = (-\Delta_h)^{-1} \q_h + C_{\q_h}^* \quad \mbox{with} \, \, \,
C_{\q_h}^* = \frac23 \Bigl( \overline{\p_h} + k \epsilon \overline{ \hat{\m}_h \times \q_h}
+ \alpha k \epsilon\overline{\hat{\m}_h \times ( \hat{\m}_h \times \q_h) } \Bigr)
\label{m-representation-1}
\end{equation*}
and $\hat{\m}_h$ is given by \cref{m_hat}. In turn, \cref{eqqq-19} is then rewritten as
\begin{equation}
G (\q_h) := \frac32 ( (-\Delta_h)^{-1} \q_h + C_{\q_h}^* ) - \p_h
- k\epsilon\hat{\m}_h \times \q_h
- \alpha k\epsilon \hat{\m}_h \times ( \hat{\m}_h \times \q_h) = \0.
\label{scheme-alt-2}
\end{equation}

To proceed, we need the following lemma.
\begin{lem}[Browder-Minty lemma \cite{browder1963nonlinear,minty1963monotonicity}] \label{lem:Browder}
	Let $\mathcal{X}$ be a real, reflexive Banach space and let $T: \mathcal{X} \to \mathcal{X}'$ (the dual space of $\mathcal{X}$) be bounded, continuous,
	coercive (i.e., $\frac{ (T (u) , u ) }{ \| u \|_{\mathcal{X}} } \to +\infty$, as $\| u \|_{\mathcal{X}} \to +\infty$) and monotone.
	Then for any $g \in {\mathcal{X}}'$ there exists a solution $u \in {\mathcal{X}}$ of the equation $T (u) =g$.
	
	Furthermore, if the operator $T$ is strictly monotone, then the solution $u$ is unique.
\end{lem}

\begin{proof}[Proof of \thmref{thm2}]
	For any $\q_{1,h}$, $\q_{2,h}$ with $\overline{\q_{1,h}} = \overline{\q_{2,h}}=0$,
	we denote $\tilde{\q}_h = \q_{1,h} - \q_{2,h}$ and derive the following monotonicity estimate:
	\begin{align*}
	&\langle G (\q_{1,h}) - G (\q_{2,h}) , \q_{1,h} - \q_{2,h} \rangle \\
	&=\frac{3}{2} \Bigl(
	\langle (-\Delta_h)^{-1} \tilde{\q}_h , \tilde{\q}_h \rangle
	+ \langle C_{\q_{1,h}}^* - C_{\q_{1,h}}^* , \tilde{\q}_h \rangle \Bigr)\nonumber\\
	& - k\epsilon\langle \hat{\m}_h \times \tilde{\q}_h , \tilde{\q}_h \rangle
	- \alpha k\epsilon \langle\hat{\m}_h \times ( \hat{\m}_h \times \tilde{\q}_h ) ,
	\tilde{\q}_h \rangle\nonumber\\
	&\ge \frac{3}{2} \Bigl(
	\langle (-\Delta_h)^{-1} \tilde{\q}_h , \tilde{\q}_h \rangle
	+ \langle C_{\q_{1,h}}^* - C_{\q_{2,h}}^* , \tilde{\q}_h \rangle \Bigr) \nonumber\\
	&= \frac{3}{2 }
	\langle (-\Delta_h)^{-1} \tilde{\q}_h , \tilde{\q}_h \rangle
	= \frac{3}{2 } \| \tilde{\q}_h \|_{-1}^2 \ge 0 . \nonumber
	\end{align*}
	Note that the following equality and inequality have been applied in the second step:
	\begin{align*}
	\langle \hat{\m}_h \times \tilde{\q}_h , \tilde{\q}_h \rangle  & = 0 , \quad
	\langle \hat{\m}_h \times ( \hat{\m}_h \times \tilde{\q}_h ), \tilde{\q}_h \rangle \le 0 .
	\end{align*}
	The third step is based on the fact that both $C_{\q_{1,h}}^*$ and $C_{\q_{2,h}}^*$ are constants, and $\overline{\q_{1,h}} = \overline{\q_{2,h}}=0$, so that $\langle C_{\q_{1,h}}^* - C_{\q_{2,h}}^* , \tilde{\q}_h \rangle = 0$.
	Moreover, for any $\q_{1,h}$, $\q_{2,h}$ with $\overline{\q_{1,h}} = \overline{\q_{2,h}}=0$, we get
	\begin{eqnarray*}
		\langle G (\q_{1,h}) - G (\q_{2,h}) , \q_{1,h} - \q_{2,h} \rangle
		\ge \frac{3}{2 k} \| \tilde{\q}_h \|_{-1}^2 > 0 ,  \quad
		\mbox{if $\q_{1,h} \ne \q_{2,h}$} , \label{solvability-3}
	\end{eqnarray*}
	and the equality only holds when $\q_{1,h} = \q_{2,h}$.
	
	Therefore, an application of \lemref{lem:Browder} implies a unique solution of \textrm{\bf Scheme A}.
	
\end{proof}
	
\begin{thm}[Solvability for \textrm{\bf Scheme B}]\label{thm3}
	Given $\tilde{\p}_h$ and $\hat{\m}_h$, the numerical scheme \cref{eqqq-20} admits a unique solution.
\end{thm}

	\begin{proof}[Proof of \thmref{thm3}]
		We first rewrite \cref{eqqq-20} in a  compact form $\left(\frac{3}{2}I-A \right)\m_h=\tilde{\p}_h$, where
		\begin{align*}
		A = &\; \frac{3}{2}\alpha\hat{\m}_h\times I_h+k\epsilon(1+\alpha^2)\hat{\m}_h\times(-\Delta_h)\\
		= &\; k\epsilon(1+\alpha^2)\hat{\m}_h\times\left( -\Delta_h+\frac{3\alpha}{2k\epsilon(1+\alpha^2)}I_h\right)\\
		=:& \; k\epsilon(1+\alpha^2) MS.	
		\end{align*}
		Here $I_h$ is the identity matrix, $M$ is the antisymmetric matrix corresponding to the discrete operator $\m_h\times$,
		and $S$ is the symmetric positive definite matrix corresponding to $-\Delta_h+\frac{3\alpha}{2k\epsilon(1+\alpha^2)}I_h$ which admits a decomposition $S=CC^T$
		with $C$ being nonsingular.
		
		Thus, we have
		\begin{align*}
		|\lambda I-MS|=|\lambda I-MC^TC|=|\lambda I-CMC^T|,
		\end{align*}
		Due to the antisymmetry of matrix $M$, we have
		\begin{align*}
		(CMC^T)^T = \; & -CMC^T. 
		\end{align*}
		It follows from the spectral lemma for antisymmetric matrices \citep{prasolov1994problems} that the eigenvalues of $CMC^T$ are either $0$ or purely imaginary, 
		thus the eigenvalues of $MS$ are either $0$ or purely imaginary as well. This unique solvability comes as a consequence of the fact that
		all eigenvalues of $\frac{3}{2}I-A$ have $\frac{3}{2}$ as real parts and $\det(\frac{3}{2}I-A)\neq 0$.
	\end{proof}

\begin{remark}
Note that the unique solvability of \textbf{Scheme A} and \textbf{Scheme B} \textit{does not} impose any condition on $k$ and $h$.
	This is in contrast with earlier results for the fully implicit schemes where $k=\mathcal{O}(h^2)$ is needed for the unique solution of the nonlinear
	system of equations at each time step; see \citep{fuwa2012finite} for example.
\end{remark}

\begin{remark}
	In \citep{jingrun2019analysis}, we prove the second-order convergence of a modified scheme of \textbf{Scheme A} by introducing two sets of approximated solutions and separating errors caused by the evolution step and by the projection step.
	A similar proof of \thmref{thm2} has also been given in \citep{jingrun2019analysis}.
\end{remark}

\subsection{Computation of the stray field}\label{section:stray}
 The stray field $\h_s=-\nabla U$ with $U$ the scalar function in $\mathbb{R}^3$ which satisfies
	\begin{align}\label{poisson_eqn}
	\Delta U = \; & \left\{ 
	\begin{aligned}
	&\quad  \nabla \cdot \m \quad \textrm{ in } \Omega \\
	&\qquad  0 \qquad\textrm{ outside } \Omega, 
	\end{aligned}
	\right.
	\end{align}
	together with jump conditions
	\begin{align}\label{poisson_boundary}
	\left[U\right]_{\partial \Omega} = \;& 0 \nonumber\\
	\left[\frac{\partial U}{\partial {\bf \nu}}\right]_{\partial \Omega} = \;& -\m \cdot {\bf \nu}.
	\end{align}
Here $[U]_{\partial \Omega}$ denotes the jump of $U$ at the material boundary as
	\begin{equation}
	[U]_{\partial \Omega}({\bm x})=\lim_{\substack{{\bm y}\to {\bm x} \\ {\bm y}\in \mathbb{R}^3/\Omega}} U({\bm y})-\lim_{\substack{{\bm y}\to {\bm x} \\ {\bm y}\in \Omega}} U({\bm y}),
	\end{equation}
	and $\left[\frac{\partial U}{\partial {\bf \nu}}\right]_{\partial \Omega}$ is defined similarly.
	The solution to \cref{poisson_eqn} - \cref{poisson_boundary} is
	\begin{align*}
	U(\x) = \int_{\Omega} \nabla N({\bm x}-{\bm y})\cdot {\bm m}({\bm y})\,d{\bm y},
	\end{align*}
	and thus the stray field
	\begin{equation}\label{eqq-5}
	{\h}_{\text{s}}=-\nabla \int_{\Omega} \nabla N({\bm x}-{\bm y})\cdot {\bm m}({\bm y})\,d{\bm y},
	\end{equation}
	where $N({\bm x})=-\frac{1}{4\pi |{\bm x}|}$ is the Newtonian potential.
	
	It follows from \cref{eqq-5} combined with the divergence theorem that
	\begin{align}\label{eqn:div}
	{\h}_s({\bm r}) = \frac{1}{4\pi}\nabla\left\{\int_{\Omega}\frac{\nabla \cdot \m({\bm r}^\prime)}{|{\bm r}-{\bm r}^{\prime}|}d{\bm r}^\prime-\int_{\partial\Omega}\frac{\m({\bm r}^\prime)\cdot {\bm n}({\bm r}^{\prime})}{|{\bm r}-{\bm r}^{\prime}|}dS({\bm r}^\prime)\right\}.
	\end{align}
	The evaluation of the stray field can be carried out by performing an integration over the entire material for every point ${\bm r}$. 
	In terms of the computational complexity, a direct evaluation of \eqref{eqq-5} requires $\mathcal{O}(N^2)$ with $N$ the
	degree of freedoms. Moreover, we need to evaluate \eqref{eqq-5} at each time step. Therefore, 
	the direct evaluation is computationally expensive and thus a fast solver is highly desirable. 
	The complexity for solving stray field using FFT is $\mathcal{O}(N\log N)$ \cite{wang2001gauss, Garcia2001Improved}.

\section{Accuracy and efficiency test}\label{section:numerical1}

We use examples in both 1D and 3D to show the second-order accuracy of \textbf{Scheme A} and \textbf{Scheme B}.

In order to have the exact magnetization profile, we consider the simplified LLG equation with only the exchange term and
set $\epsilon=1$ with the forcing term
\begin{align}\label{eq-24}
\m_t=-\m\times\Delta\m-\alpha\m\times(\m\times\Delta\m) + \g,
\end{align}
where $\g=\m_{et}+\m_e\times\Delta \m_e+\alpha \m_e\times(\m_e\times\Delta \m_e)$ with $\m_e$ the exact solution. The ferromagetic body $\Omega = [0,1]$ in 1D and $\Omega=[0,1]^3$ in 3D.
The final time $T=1$. Since the exchange term is the stiffest term in the original LLG equation, it is adequate to
use \eqref{eq-24} to test accuracy and efficiency of the proposed methods.

The exact solution in 1D is
\begin{equation}\label{eqn:exact1d}
\m_e=\left(\cos(x^2(1-x)^2)\sin t, \sin(x^2(1-x)^2)\sin t, \cos t\right)^T,
\end{equation}
which satisfies the homogeneous Neumann boundary condition.

The exact solution in 3D is
\begin{align}\label{eqn:exact3d}
	\m_e=\left(\cos(XYZ)\sin t, \sin(XYZ)\sin t, \cos t\right)^T,
\end{align}
where $X=x^2(1-x)^2$, $Y=y^2(1-y)^2$, and $Z=z^2(1-z)^2$.

Since both schemes are semi-implicit, we compare their efficiency with another
two semi-implicit methods: Gauss-Seidel projection method~\cite{wang2001gauss} and
the second-order implicit-explicit method~\citep{Boscarino2016High}. For completeness, we first
state these two methods.

\subsection{The Gauss-Seidel projection method}

It follows from \cref{eq-4,eq-5} that the Gauss-Seidel projection method (GSPM) ~\cite{wang2001gauss} is
given as the following three steps:
\begin{enumerate}[Step 1. ]
	\item Implicit Gauss-Seidel:
	\begin{align}
	g_i^n &= (I-\epsilon \Delta t\Delta_h)^{-1}(m_i^n+\Delta tf_i^n),\nonumber \\
	g_i^{*} &= (I-\epsilon \Delta t\Delta_h)^{-1}(m_i^{*}+\Delta tf_i^{n}),\ \ i=1,2,3  \label{eq-17}
	\end{align}
	{\begin{equation}
		\begin{pmatrix}
		m_1^{*}\\
		m_2^{*}\\
		m_3^{*}
		\end{pmatrix}
		=
		\begin{pmatrix}
		m_1^n+(g_2^nm_3^n-g_3^nm_2^n)\\
		m_2^n+(g_3^nm_1^{*}-g_1^{*}m_3^n)\\
		m_3^n+(g_1^{*}m_2^{*}-g_2^{*}m_1^{*})
		\end{pmatrix}.
		\end{equation}}
	\item Heat flow without constraints :
	\begin{equation}\label{eq-19}
	{\bm f}^{*}=-Q(m_2^{*}{\bm e}_2+m_3^{*}{\bm e}_3)+{\bm h}_s^{n}+{\bm h}_e
	\end{equation}
	
	{\begin{equation}
		\begin{pmatrix}
		m_1^{**}\\
		m_2^{**}\\
		m_3^{**}
		\end{pmatrix}
		=
		\begin{pmatrix}
		m_1^{*}+\alpha \Delta t (\epsilon \Delta_hm_1^{**}+f_1^{*})\\
		m_2^{*}+\alpha \Delta t (\epsilon \Delta_hm_2^{**}+f_2^{*})\\
		m_3^{*}+\alpha \Delta t (\epsilon \Delta_hm_3^{**}+f_3^{*})
		\end{pmatrix}.
		\end{equation}}
	\item Projection onto $S^2$:
	{\begin{equation}
		\begin{pmatrix}
		m_1^{n+1}\\
		m_2^{n+1}\\
		m_3^{n+1}
		\end{pmatrix}
		=
		\frac{1}{|m^{**}|}\begin{pmatrix}
		m_1^{**}\\
		m_2^{**}\\
		m_3^{**}
		\end{pmatrix}.
		\end{equation}}
\end{enumerate}
where $\m^*$ denotes the intermediate values of $\m$. As in \cite{Garcia2001Improved}, the stray field $\h_s$ is computed using the intermediate values $\m^{*}$ in \cref{eq-17} and \cref{eq-19}.

\subsection{The second-order implicit-explicit method}

For simplicity, we only formulate the second-order implicit-explicit (IMEX2) scheme \citep{Boscarino2016High}
without Gilbert damping.
Define
\begin{equation*}
{\bm H}(t,{\bm u},{\bm v})=-{\bm u}\times \Delta {\bm v} + {\bm f}
\end{equation*}
and
\begin{align*}
{\bm u^1} &= {\bm m}^n, \\
{\bm v}^1 &= {\bm m}^n.
\end{align*}
The second-order IMEX scheme (IMEX2) reads as
\begin{align}
{\bm \ell}^1 &= {\bm H}(t^n+\gamma k,{\bm u}^1,{\bm v}^1+\gamma k{\bm \ell}^1), \nonumber \\
{\bm k}^1 &= {\bm H}(t^n,{\bm u}^1,{\bm v}^1+\gamma k{\bm \ell}^1),\nonumber \\
{\bm u}^2 &= {\bm u}^1 + k{\bm k^1}, \nonumber \\
{\bm v}^2 &= {\bm u}^1 + (1-2\gamma)k {\bm k}^1, \nonumber \\
{\bm \ell}^2 &= {\bm H}(t^n + (1-\gamma)k,{\bm u}^2,{\bm v}^2+\gamma k {\bm \ell}^2),  \nonumber \\
{\bm k}^2 &= {\bm H}(t^n + k,{\bm u}^2,{\bm v}^2+\gamma k{\bm \ell}^2), \nonumber \\
\tilde{{\bm m}}^{n+1} &= {\bm m}^n +\frac{1}{2} k({\bm k}^1 + {\bm k}^2),\nonumber \\
{\bm m}^{n+1} &= \frac{\tilde{{\bm m}}^{n+1}}{|\tilde{{\bm m}}^{n+1}|},\label{IMEX2}
\end{align}	
where $\gamma=1-1/\sqrt{2}$.  Note that IMEX2 solves two linear systems of equations at each step.

\begin{remark}
The computational cost for GSPM, IMEX2, and BDF2 at each temporal step is as follows.
Seven symmetric linear systems of equations with constant coefficients and dimension $M$ need to be solved for GSPM,
and two nonsymmetric linear systems of equations with variable coefficients and dimension $3M$ need to be solved for IMEX2.
BDF2 only needs to solve one nonsymmetric linear systems of equations with variable coefficients and dimension $3M$ for 
both \textbf{Scheme A} and \textbf{Scheme B}. Here $M$ is the number of unknowns in each spatial dimension.
\end{remark}

\subsection{Accuracy test}

Since \textrm{\bf Scheme A} and \textrm{\bf Scheme B} are comparable numerically, we only show results of
\textrm{\bf Scheme A}, termed as BDF2. For comparison, we also list results of the other two semi-implicit 
methods: GSPM and IMEX2.

In 1D, for \eqref{eqn:exact1d}, we fix $h=5D-4$ and record the temporal error in terms of the temporal stepsize $k$ in \Cref{tab-2}
and \Cref{all_norms_time_1D} to get the temporal accuracy. Both BDF2 and IMEX2 are second-order accurate,
while GSPM is first-order accurate. To get the spatial accuracy, 
we fix $k=1D-6$ and record the spatial error in terms of $h$ in \Cref{tab-5} and \Cref{all_norms_space_1D}.
BDF2, IMEX2, and GSPM are all second-order accurate.

\begin{table}[htbp]
	\centering
	\caption{Temporal accuracy in 1D for BDF2, GSPM, and IMEX2 when $h=5D-4$ and $\alpha=1D-3$.}\label{tab-2}
	\begin{tabular}{|c|p{2.5cm}<{\centering}|p{2.5cm}<{\centering}|p{2.5cm}<{\centering}|}
		\hline
		{}  & BDF2 & GSPM &IMEX2 \\
		$k$ & $\|m_h-m_e\|_{\infty}$ & $\|m_h-m_e\|_{\infty}$ & $\|m_h-m_e\|_{\infty}$\\
		\hline
		2.0D-2 &1.0753D-4 & 3.2024D-2 & 4.1119D-5 \\
		1.0D-2 & 2.7384D-5 &1.6289D-2  &1.3107D-5  \\	
		5.0D-3 & 6.8538D-6 &7.9890D-3 &3.7623D-6  \\
		2.5D-3 & 1.6513D-6 & 4.1923D-3 & 7.2361D-7 \\
		1.25D-3 &3.4152D-7 & 2.0662D-3&1.2711D-7 \\
		\hline
		order & 2.065 & 0.987 & 2.085 \\
		\hline
	\end{tabular}
\end{table}

\begin{table}[htbp]
	\centering
	\caption{Spatial accuracy in 1D for BDF2, GSPM, and IMEX2 when $k=1D-6$ and $\alpha=1D-3$.} \label{tab-5}
	\begin{tabular}{|c|p{2.5cm}<{\centering}|p{2.5cm}<{\centering}|p{2.5cm}<{\centering}|}
		\hline
		{}  & BDF2 & GSPM & IMEX2 \\
		$h$ & $\|m_h-m_e\|_{\infty}$ & $\|m_h-m_e\|_{\infty}$ & $\|m_h-m_e\|_{\infty}$ \\
		\hline
		4.0D-2 & 6.2092D-4 & 6.2094D-4 & 6.2092D-4 \\
		2.0D-2  & 1.5516D-4 & 1.5517D-4 & 1.5516D-4 \\
		1.0D-2  & 3.8789D-5 & 3.8805D-5 & 3.8789D-5 \\
		5.0D-3  & 9.6973D-6 & 9.7145D-6 & 9.6972D-6 \\
		2.5D-3  & 2.4243D-6 & 2.4439D-6 & 2.4243e-6 \\
		\hline
		order  & 2.000 &1.998 & 2.000 \\
		\hline
	\end{tabular}
\end{table}

In 3D, for \eqref{eqn:exact3d}, we fix $h_x=h_y=h_z=1/16$ and record the temporal error in terms of $k$ in \cref{tab-6} and \cref{all_norms_time_3D} to get the temporal accuracy. Both BDF2 and IMEX2 are second-order accurate, while GSPM is first-order accurate.
To get the spatial accuracy, we fix the temporal stepsize $k=1D-4$ and record the spatial error in terms of $h_x=h_y=h_z=h$ in \cref{tab-7} and \cref{all_norms_space_3D}. BDF2, IMEX2, and GSPM are all second-order accurate.
\begin{table}[htbp]
	\centering
	\caption{Temporal accuracy in 3D for BDF2, GSPM, and IMEX2 when $h_x=h_y=h_z=1/16$ and $\alpha=1D-3$.}\label{tab-6}
	\begin{tabular}{|c|p{2cm}<{\centering}|c|p{2cm}<{\centering}|c|p{2cm}<{\centering}|}
		\hline
		{}  & BDF2& {} & GSPM & {} &IMEX2 \\
		$k$ & $\|m_h-m_e\|_{\infty}$& $k$ & $\|m_h-m_e\|_{\infty}$ & $k$ & $\|m_h-m_e\|_{\infty}$\\
		\hline
		1/8 &3.1853D-3 & 1/32&3.0024D-4 & 1/4& 4.1759D-3\\
		1/16 & 9.3167D-4 &1/64 &1.5214D-4 & 1/8& 1.0659D-3 \\	
		1/32 & 2.4929D-4 & 1/128& 7.7368D-5&1/16 & 2.6611D-4\\
		1/64 & 6.1027D-5 & 1/256 & 3.8171D-5& 1/32& 6.3198D-5\\
		1/128 &1.1785D-5 & 1/512& 2.0295D-5& 1/64&1.2107D-5\\
		\hline
		order & 2.009 & - & 0.977 & - &2.094 \\
		\hline
	\end{tabular}
\end{table}

\begin{table}[htbp]
	\centering
	
	\caption{Spatial accuracy in 3D for BDF2, GSPM, and IMEX2 when $k=1D-4$ and $\alpha=1D-3$.} \label{tab-7}
	\begin{tabular}{|c|p{2.5cm}<{\centering}|p{2.5cm}<{\centering}|p{2.5cm}<{\centering}|}
		\hline
		{}  & BDF2 & GSPM & IMEX2 \\
		$h$ & $\|m_h-m_e\|_{\infty}$ & $\|m_h-m_e\|_{\infty}$ & $\|m_h-m_e\|_{\infty}$ \\
		\hline
		1/2 & 3.5860D-4 &3.5856D-4  &3.5860D-4  \\
		1/4 &8.1251D-5 & 8.1245D-5&8.1252D-5  \\
		1/8 & 2.0187D-5 & 2.0187D-5 & 2.0189D-5 \\
		1/16  & 5.0532D-6 & 5.0519D-6 & 5.0552D-6 \\
		1/32  & 1.2635D-6 & 1.3278D-6 & 1.2656D-6 \\
		\hline
		order  & 2.030 &2.016 & 2.030\\
		\hline
	\end{tabular}
\end{table}

\begin{figure}[htbp]
\centering
\subfloat[Temporal accuracy in 1D]{\label{all_norms_time_1D}\includegraphics[width=2.5in]{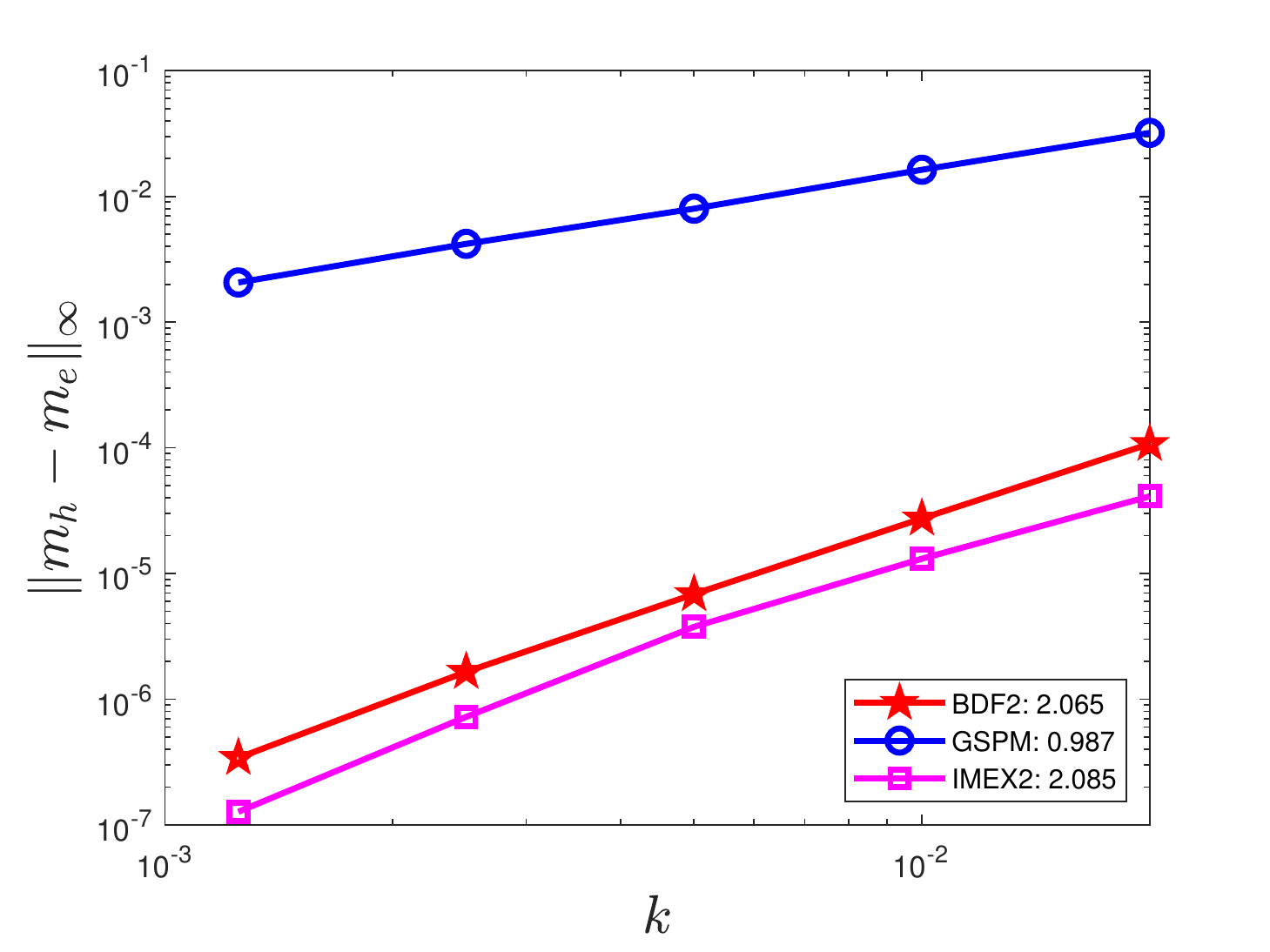}}
\subfloat[Spatial accuracy in 1D]{\label{all_norms_space_1D}\includegraphics[width=2.5in]{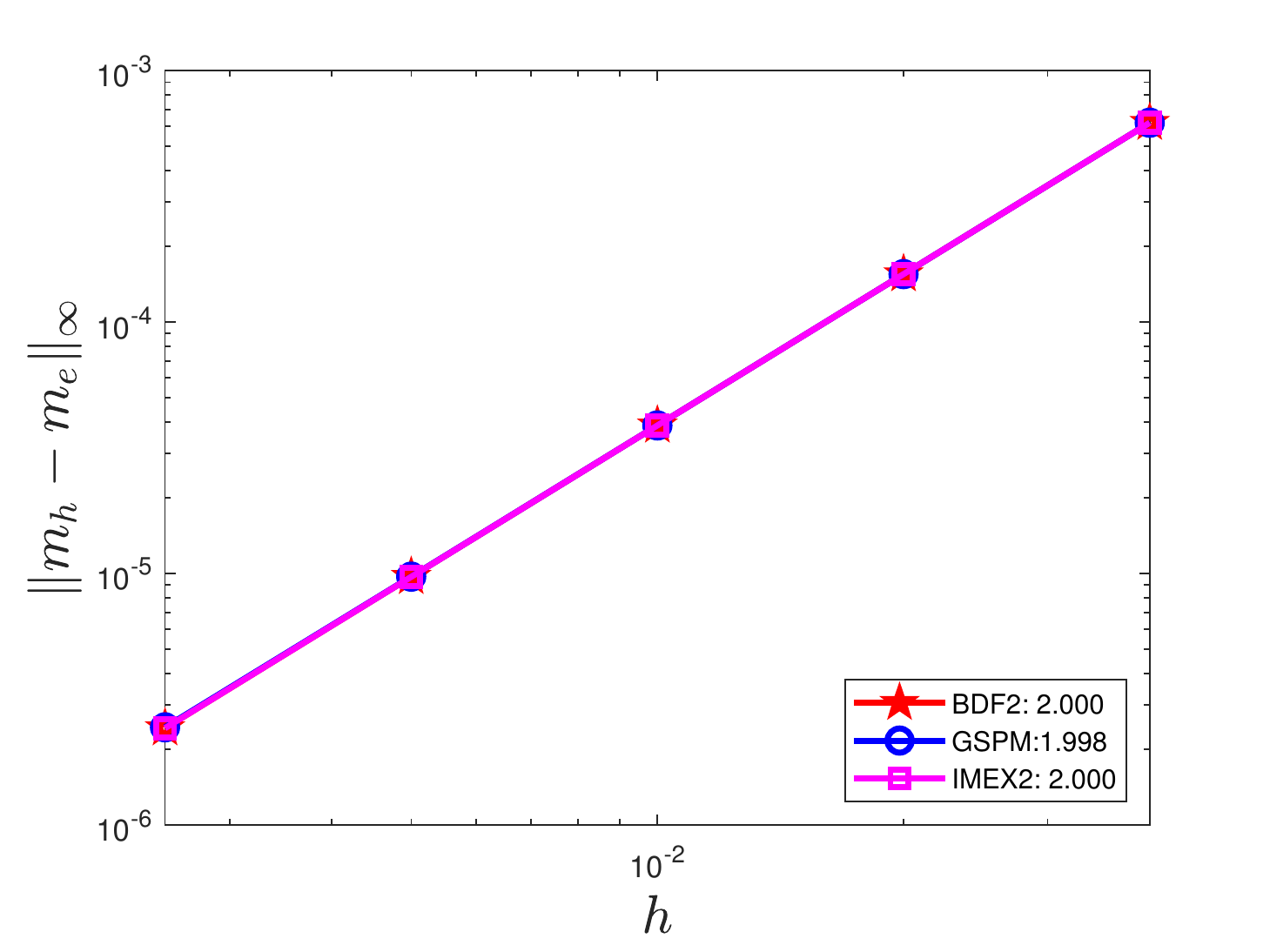}}
\caption{Accuracy test in 1D for BDF2, GSPM, and IMEX2 when $\alpha = 1D-3$. (a) Temporal accuracy when $h=5D-4$; (b) Spatial accuracy when $k=1D-6$.}\label{all_norms_1D}
\end{figure}

\begin{figure}[htbp]
\centering
\subfloat[Temporal accuracy in 3D]{\label{all_norms_time_3D}\includegraphics[width=2.5in]{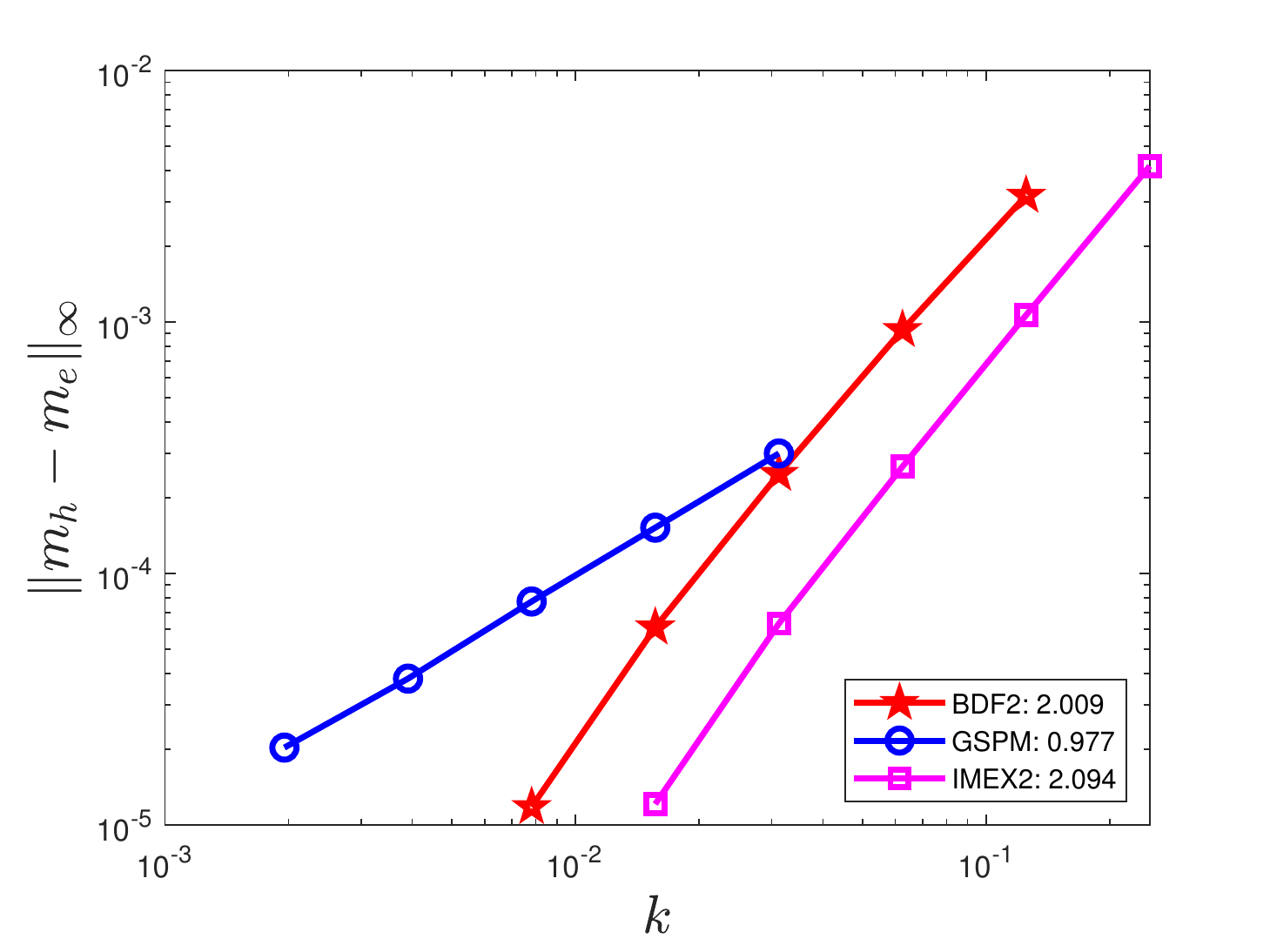}}
\subfloat[Spatial accuracy in 3D]{\label{all_norms_space_3D}\includegraphics[width=2.5in]{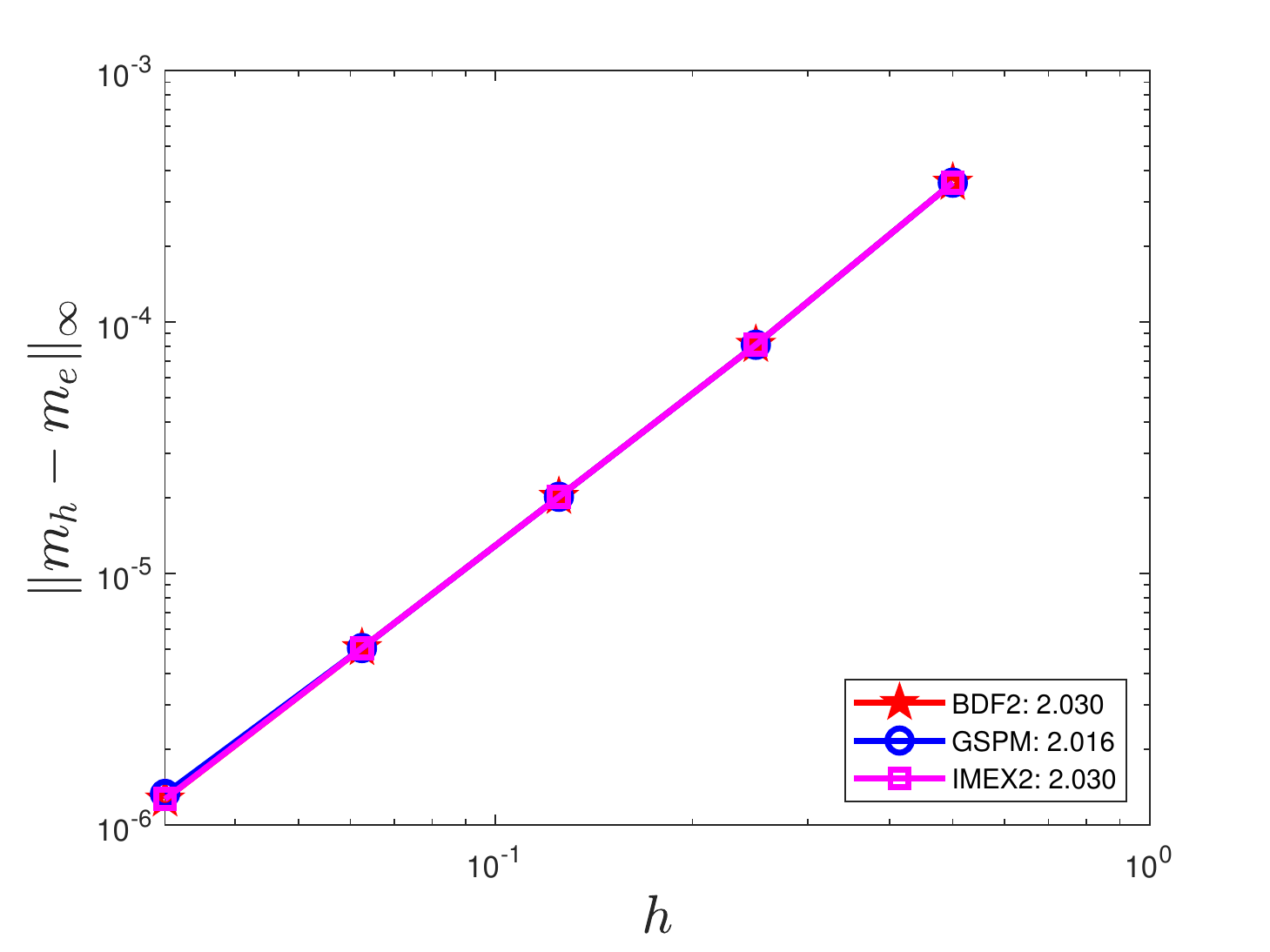}} 
\caption{Accuracy test in 3D for BDF2, GSPM, and IMEX2 when $\alpha = 1D-3$. (a) Temporal accuracy when $h_x=h_y=h_z=1/16$; 
(b) Spatial accuracy when $k=1D-4$.}\label{all_norms_3D}
\end{figure}

\subsection{Efficiency comparison}

To compare the efficiency, we plot the CPU time (in seconds) of BDF2, GSPM and IMEX2 in terms of the error $\|\m_h-\m_e\|_{\infty}$
in \cref{cpu_time_damping_1D} for the 1D case and in \cref{cpu_time_damping_3D} for the 3D case. 

In 1D, for the same tolerance, costs of BDF2, GSPM, and IMEX2 in \cref{cpu_time_damping_1D} satisfy: BDF2$<$IMEX2$<$GSPM.
In 3D, for the same tolerance, costs of BDF2, GSPM, and IMEX2 in \cref{cpu_time_damping_3D} satisfy: BDF2$\approx$IMEX2$<$GSPM.
For both cases, BDF2 is slightly better than IMEX2 since two linear systems of equations need to be solved in IMEX2 while
only one linear system needs to solved in BDF2. Both BDF2 and IMEX2 are better than GSPM.
\begin{figure}[htbp]
\centering
\subfloat[1D]{\label{cpu_time_damping_1D}\includegraphics[width=2.5in]{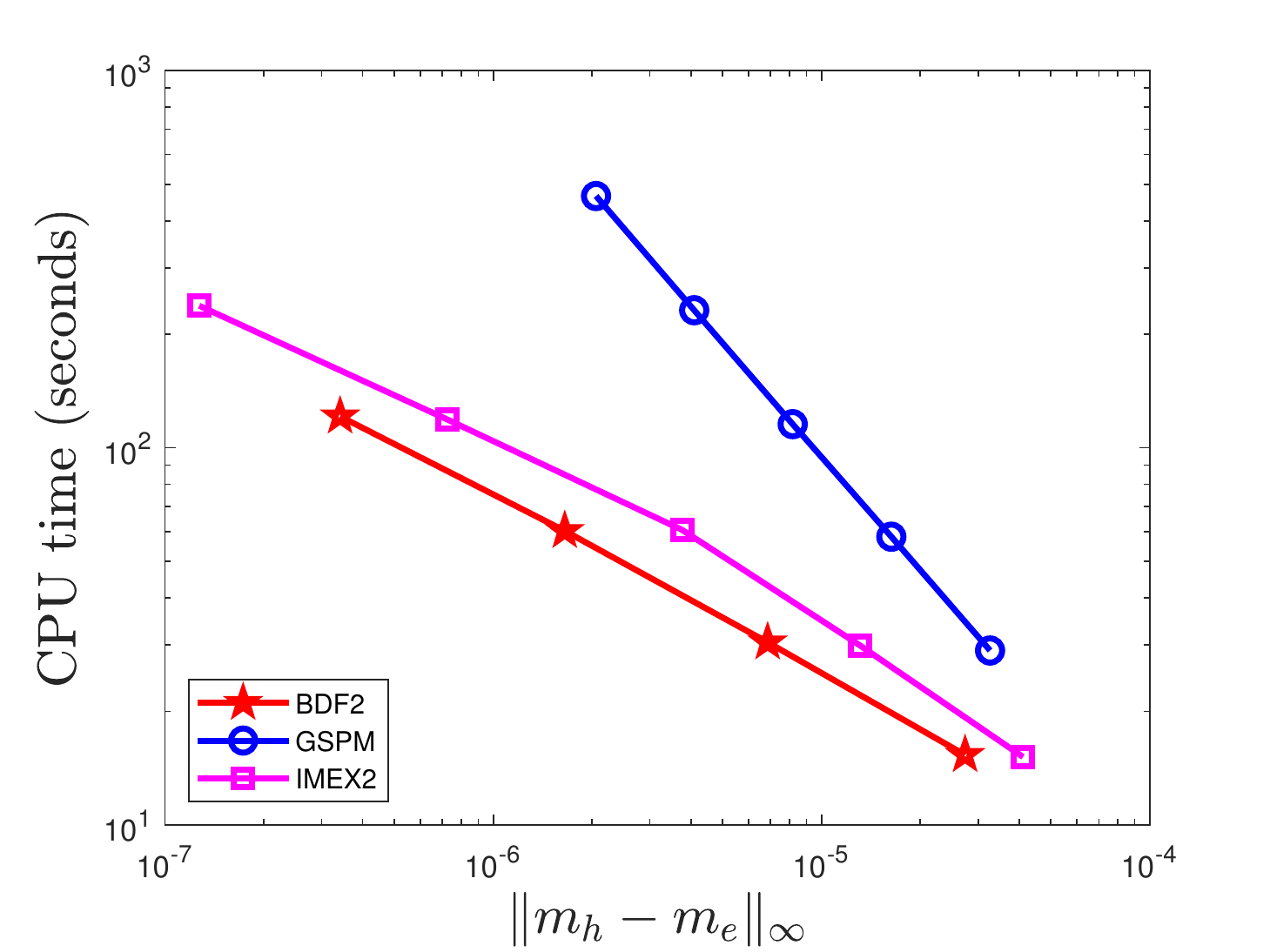}}
\subfloat[3D]{\label{cpu_time_damping_3D}\includegraphics[width=2.5in]{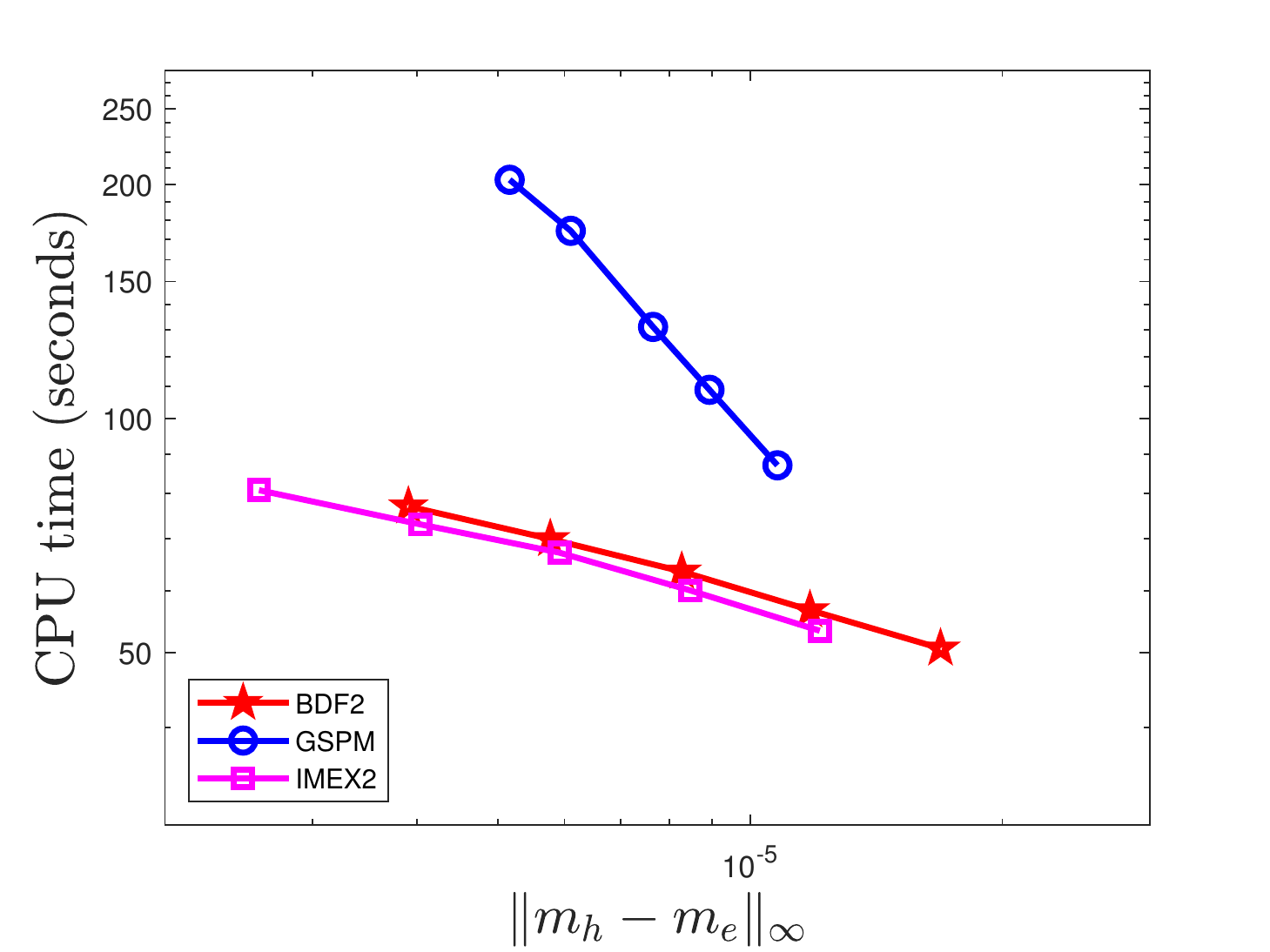}}
\caption{Efficiency test for BDF2, GSPM, and IMEX2 when $\alpha = 1D-3$ and the spatial gridsize is fixed.  
	(a) BDF2$<$IMEX2$<$GSPM when $h = 5D-4$ in 1D; (b) BDF2$\approx$ IMEX2$<$GSPM when $h_x=h_y=h_z=1/16$ in 3D.}\label{fig:efficiency}
\end{figure}

\begin{remark}
When the spatial mesh is very fine, we observe that to achieve the same tolerance, costs of BDF2, GSPM, and IMEX2 satisfy: GSPM$<$BDF2$<$IMEX2.
The reason is that fast solvers can solve symmetric linear systems with constant coefficients in GSPM, while nonsymmetric linear systems
with variable coefficients are involved in both BDF2 and IMEX2. It becomes increasingly difficult to solve such systems using the Generalized 
Minimum Residual Method (GMRES), for example. This issue will be further explored in a subsequent work. 
\end{remark}

\section{Benchmark problem from NIST}\label{section:numerical2}

To examine our methods in the realistic case, we simulate the first standard problem established by the micromagnetic modelling activity group at 
National Institute of Standards and Technology (NIST) \cite{NISTmicomagnetics}. This problem asks for simulating the hysteresis loop of a 
$L_x\times L_y\times L_z=1 \times 2\times 0.02\; \mathrm{\mu m^3}$ thin-film element with material parameters that are not too different from Permalloy. 
The hysteresis loop is obtained in the following way: A positive external field of strength $H_0=\mu_0 H_e$, in the unit of $mT$ is applied. 
The magnetization is able to reach a steady state. Once this steady state is approached, the applied external field is reduced by a certain amount, 
and the material sample is again allowed to reach a steady state. The process continues until the hysteresis system attains a negative field of strength $H_0$. 
The process then is repeated, increasing the field in small steps until it reaches the initial applied external field. 
As a consequence, we are able to plot the average magnetization at the steady state as a function of the external filed strength during the hysteresis loop. 
For BDF2, the unsymmetric linear systems of equations are solved by GMRES from library called high performance preconditioners \cite{hypreSolver} which was developed by Center for Applied Scientific Computing, Lawrence Livermore National Laboratory.

\subsection{Magnetization profile}

For comparison, we use the same setup of the available code \texttt{mo96a} of the first standard problem from NIST. 
Its setup is $100\times 50\times 1$ grid points and the canting angle $+1\degree$ of applied field from nominal axis. 
The calculation of the demagnetization field is done by FFT. The initial state is uniform. In the loop, $133$ successive
steps between $+50\;\mathrm{mT}$ and $-50\;\mathrm{mT}$ are adopted for both \textit{x-loop} and \textit{y-loop}.

Due to the presence of meta-stable symmetric states, the applied fields should be rotated one degree counterclockwise off the nominal axis. 
The damping coefficient $\alpha=0.1$, the temporal stepsize $k=1\;\mathrm{ps}$
and the cell size is $20\times 20\times 20\;\mathrm{nm^3}$. 
A stopping criterion is used to determine that a steady state is reached when the relative change in the total energy is less than $10^{-7}$. 
For \texttt{mo96a}, \cref{NIST_long_magnetization} and \cref{NIST_short_magnetization} plot the average remanent magnetization on the bottom surface of
the sample in the \textit{xy-} plane when $H_0=0$ when the applied fields are approximately parallel (canting angle $+1\degree$) to 
the \textit{y-} (long) axis and the \textit{x-} (short) axis, respectively. 
For BDF2, the corresponding results are shown in \cref{BDF2_long_magnetization} and \cref{BDF2_short_magnetization}, respectively.
The in-plane magnetization components are represented by arrows in \cref{NIST_BDF2_magnetization}. 
 
Furthermore, magnetization components are visualized by the grayscale value in \cref{NIST_BDF2_mag}. 
For \texttt{mo96a}, \cref{NIST_long_mx} - \cref{NIST_short_my} plot the \textit{x-} component and the \textit{y-} component 
when the applied field is along the \textit{y-} axis, the \textit{x-} component and
the \textit{y-} component when the applied field is along the \textit{x-} axis, respectively.
Results of BDF2 are shown in \Cref{BDF_long_mx,BDF_long_my,BDF_short_mx,BDF_short_my}.

From \cref{NIST_BDF2_magnetization} and \cref{NIST_BDF2_mag}, we observe that results of BDF2 are in qualitative agreements with those of \texttt{mo96a}.

\begin{figure}[htbp]
	\flushright
	\subfloat[Applied field parallel to the long axis]{\label{NIST_long_magnetization}\includegraphics[width=2.0in]{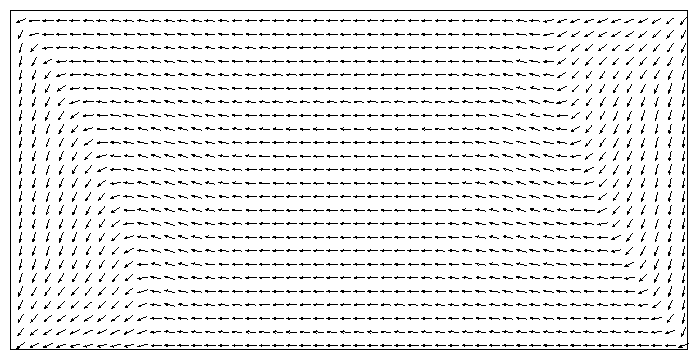}}
	\hspace{0.45in}
	\subfloat[Applied field parallel to the short axis]{\label{NIST_short_magnetization}\includegraphics[width=2.0in]{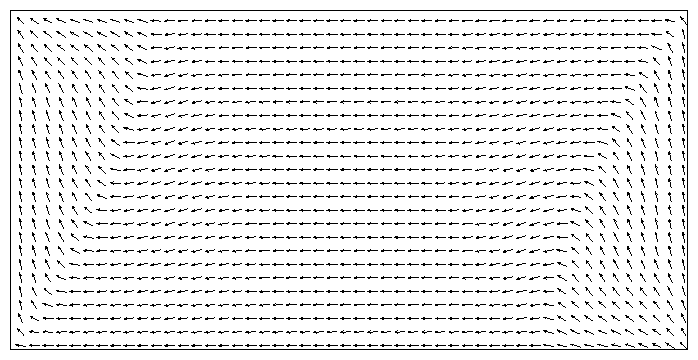}}
	\quad
	\subfloat[Applied field parallel to the long axis]{\label{BDF2_long_magnetization}\includegraphics[width=2.5in]{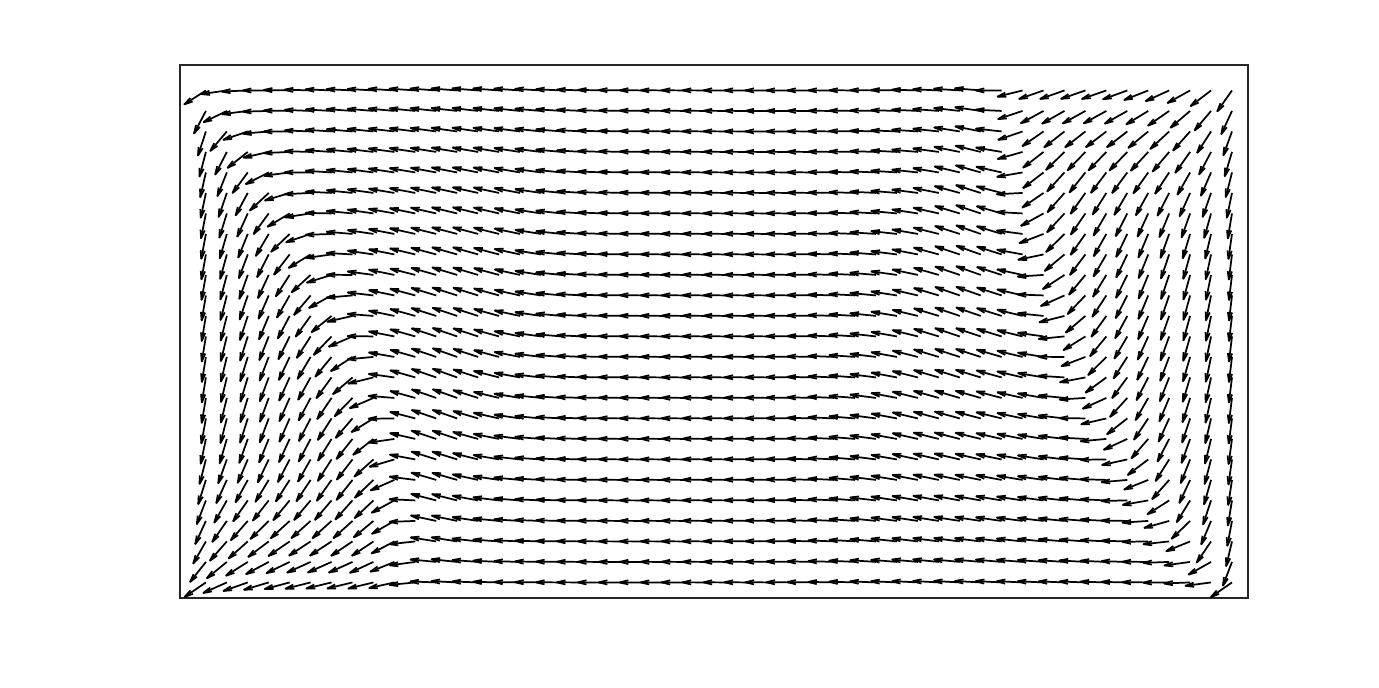}}
	\subfloat[Applied field $H_0$ parallel the short axis]{\label{BDF2_short_magnetization}\includegraphics[width=2.5in]{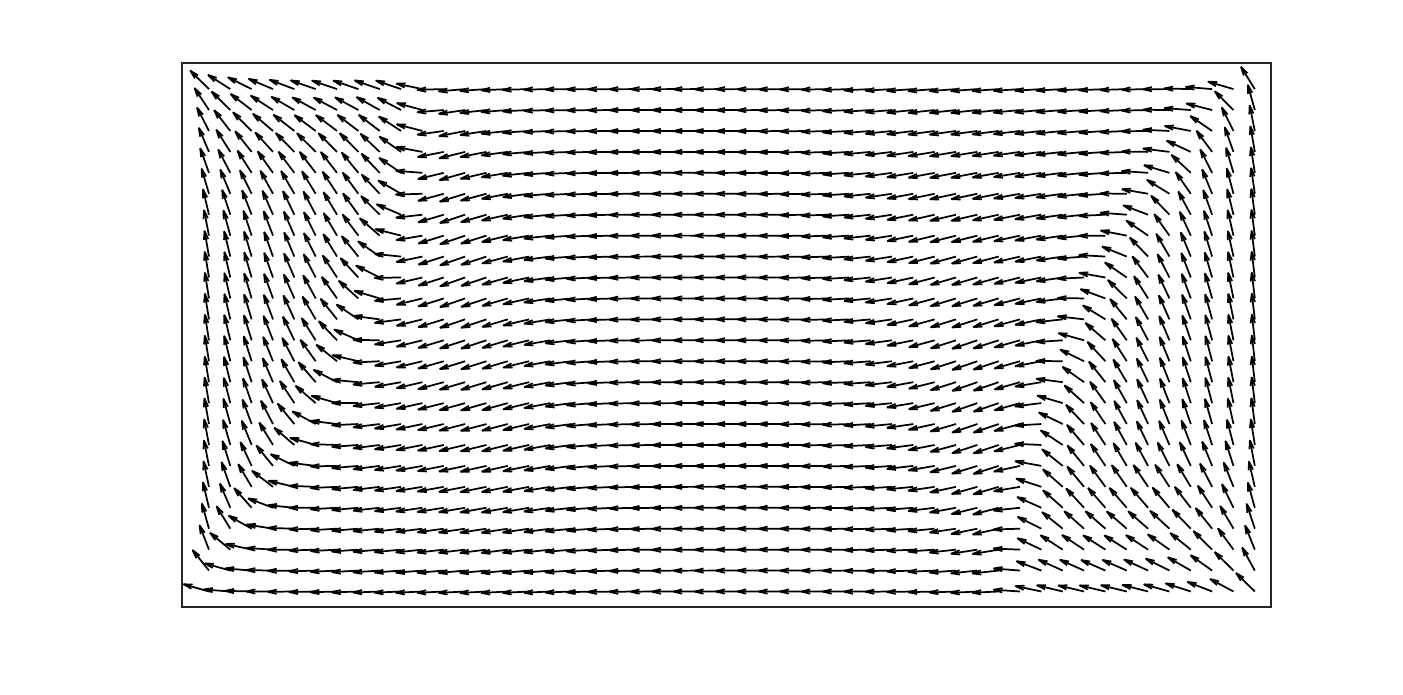}}
	\caption{Remanent magnetization when $\alpha=0.1$ for the bottom surface in the $xy$ plane. The applied field is approximately parallel (canting angle $+1\degree$) to the \textit{y-} (long) axis (left column) and the \textit{x-} (short) axis (right column). Top row: \texttt{mo96a}; Bottom row: BDF2. The in-plane magnetization components are represented by arrows.}
	\label{NIST_BDF2_magnetization}
\end{figure}

\begin{figure}[htbp]
	\centering
	\subfloat[$m_x\; (H_0 // \textit{y-axis})$ ]{\label{NIST_long_mx}\includegraphics[width=0.72in]{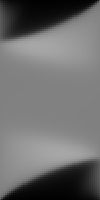}}
	\hspace{0.21in}
	\subfloat[$m_y\; (H_0 // \textit{y-axis})$ ]{\label{NIST_long_my}\includegraphics[width=0.72in]{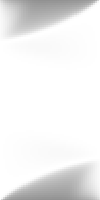}}
	\hspace{0.21in}
	\subfloat[$m_x\; (H_0 // \textit{x-axis})$ ]{\label{NIST_short_mx}\includegraphics[width=0.72in]{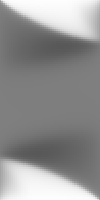}}
	\hspace{0.21in}
	\subfloat[$m_y\;(H_0 // \textit{x-axis})$ ]{\label{NIST_short_my}\includegraphics[width=0.72in]{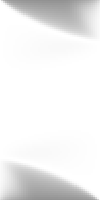}}
	\quad
	\subfloat[$m_x\; (H_0 // \textit{y-axis})$ ]{\label{BDF_long_mx}\includegraphics[width=1in]{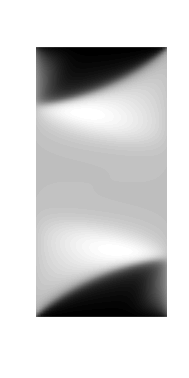}}
	\subfloat[$m_y\; (H_0 // \textit{y-axis})$ ]{\label{BDF_long_my}\includegraphics[width=1in]{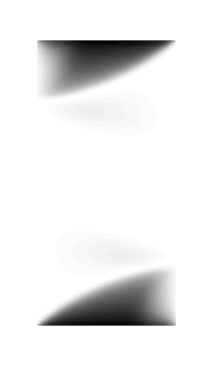}}
	\subfloat[$m_x\; (H_0 // \textit{x-axis})$ ]{\label{BDF_short_mx}\includegraphics[width=1in]{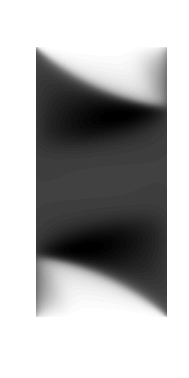}}
	\subfloat[$m_y\; (H_0 // \textit{x-axis})$ ]{\label{BDF_short_my}\includegraphics[width=1in]{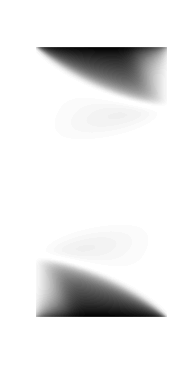}}
	\caption{ Remanent magnetization when $\alpha=0.1$ for the bottom surface in the $xy$ plane. The applied field is approximately parallel (canting angle $+1\degree$) to the \textit{y-} (long) axis (left two columns) and the \textit{x-} (short) axis (right two columns). Top row: \texttt{mo96a}; Bottom row: BDF2. The \textit{x-} and \textit{y-} magnetization components are visualized by the gray value.}
	\label{NIST_BDF2_mag}
\end{figure}

\subsection{Hysteresis loop}

Hysteresis loops generated by the code \texttt{mo96a} are shown in \cref{NIST_long_loop} 
when the applied field is approximately parallel to the long axis and in \cref{NIST_short_loop} 
when the applied field is approximately parallel to the short axis, respectively.
The average remanent magnetization in reduced units is $(-0.15,0.87,0.00)$ for the \textit{y-loop}
and $(0.15,0.87,0.00)$ for the \textit{x-loop}. The coercive fields are $4.9\;\mathrm{mT}$ 
in \cref{NIST_long_loop} and $2.5\;\mathrm{mT}$ in \cref{NIST_short_loop}.

For BDF2, hysteresis loops are presented in \cref{BDF2_long_loop} when the applied field is 
approximately parallel to the long axis and in \cref{BDF2_short_loop} when the applied field 
is approximately parallel to the short axis, respectively.
The average remanent magnetization in reduced units is $(-1.613\times10^{-1},8.606\times10^{-1},-9.940\times 10^{-5})$ 
for the \textit{y-loop} and $(1.681\times 10^{-1},8.571\times 10^{-1},-2.281\times 10^{-3})$
for the \textit{x-loop}. The coercive fields are $5.213\,(\pm 0.4) \;\mathrm{mT}$ in \cref{BDF2_long_loop} 
and $2.552\,(\pm 0.4)\;\mathrm{mT}$ in \cref{BDF2_short_loop}.

Based on these results, we conclude that results of BDF2 agree well with those of NIST, both qualitatively and quantitatively.

\begin{figure}[htbp]
	\centering
	\subfloat[$H_0 // \textit{y-axis}$]{\label{NIST_long_loop}\includegraphics[width=2.4in]{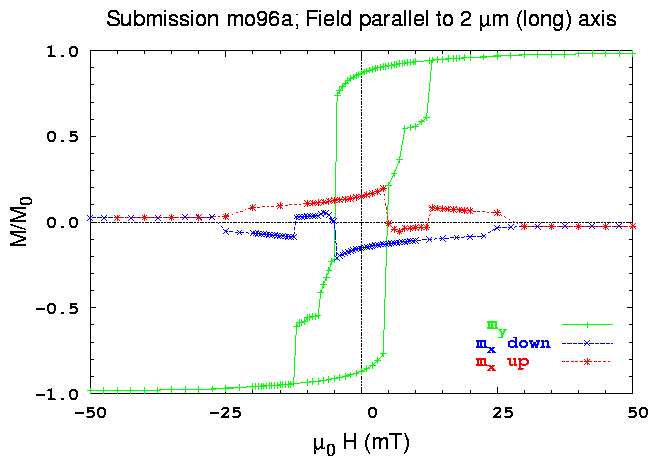}}
	\subfloat[$H_0 // \textit{x-axis}$]{\label{NIST_short_loop}\includegraphics[width=2.4in]{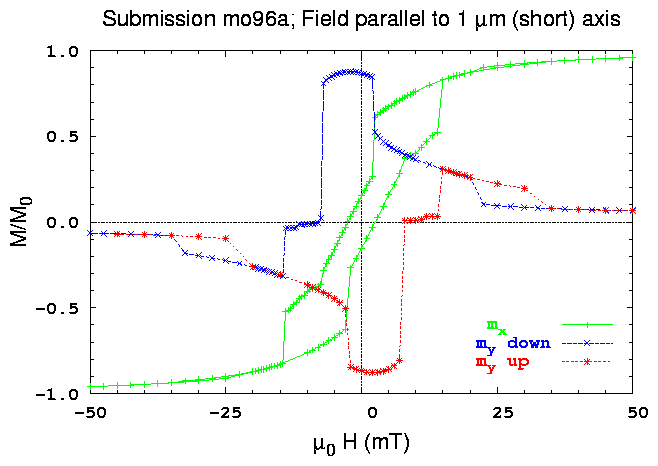}}
	\quad
	\subfloat[$H_0 // \textit{y-axis}$]{\label{BDF2_long_loop}\includegraphics[width=2.5in]{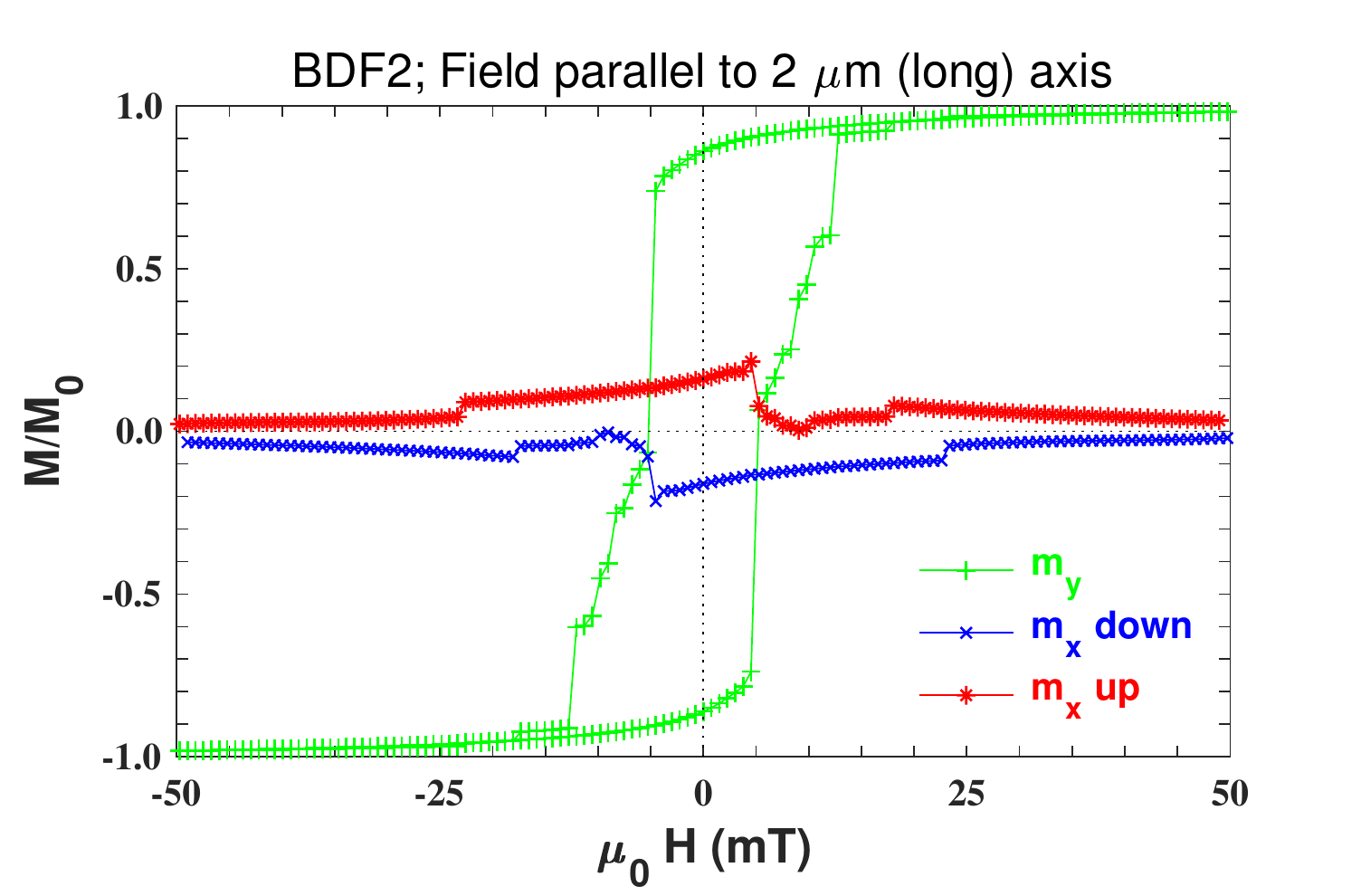}}
	\subfloat[$H_0 // \textit{x-axis}$]{\label{BDF2_short_loop}\includegraphics[width=2.5in]{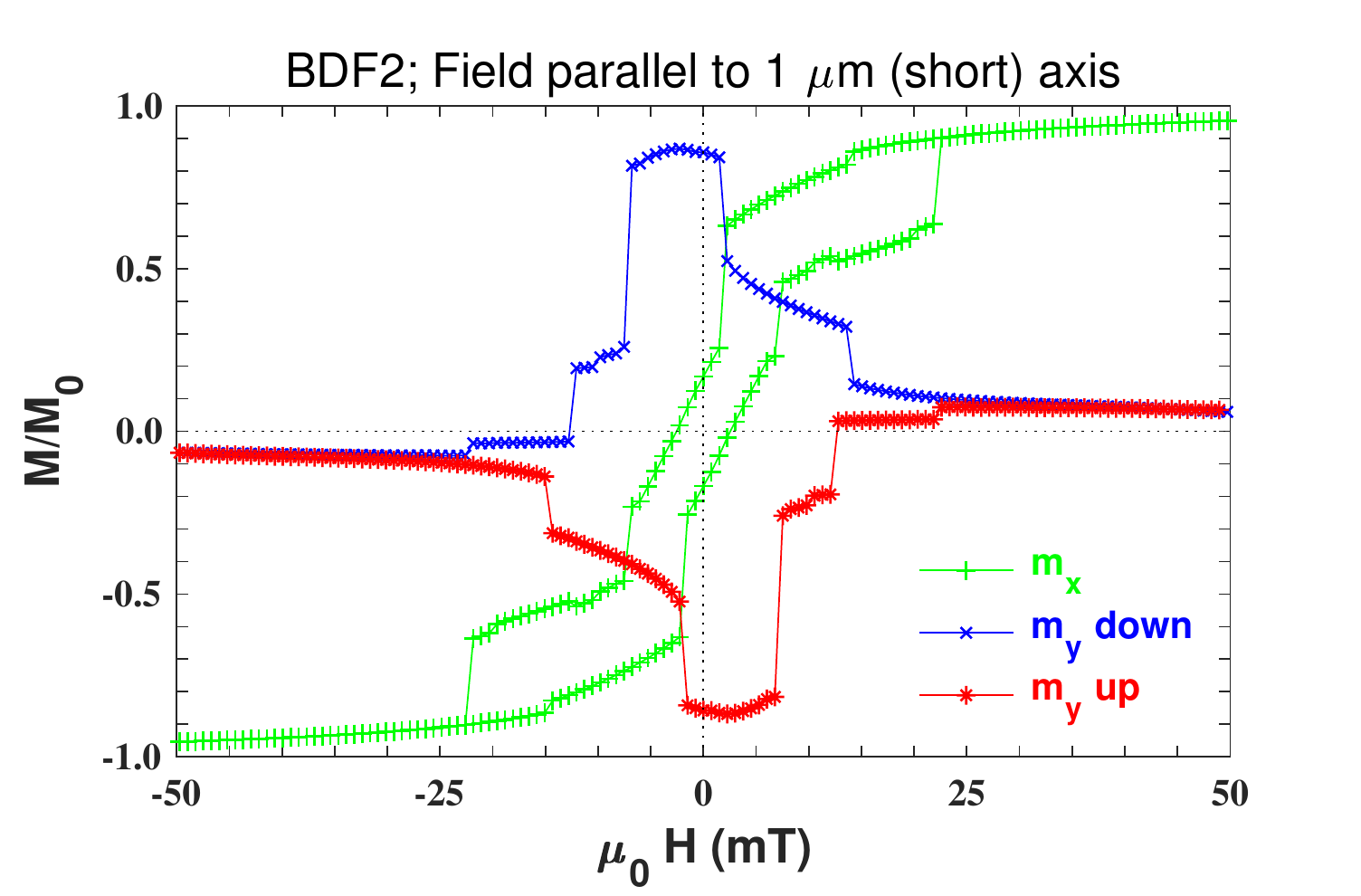}}
	\caption{ Hysteresis loop when $\alpha=0.1$ and the cell size is $20\;\textrm{nm}\times 20 \;\textrm{nm}\times 20\;\textrm{nm}$. 
The applied field is approximately parallel (canting angle $+1\degree$) to the \textit{y-} (long) axis (left column) and the \textit{x-} (short) axis (right column). Top row: \texttt{mo96a}; Bottom row: BDF2.}
	\label{GSPM}
\end{figure}

\section{Conclusions}\label{section:conclusion}

In this paper, we proposed two second-order semi-implicit projection methods to solve the Landau-Lifshitz-Gilbert equation, 
which possess unconditionally unique solvability at each time step. Examples in both 1D and 3D are used to verify the accuracy 
and the efficiency of both schemes. Additionally, the first benchmark problem from NIST is used to check the applicability of our
methods in the realistic situation using the full Landau-Lifshitiz-Gilbert equation. Results show that our schemes can produce
the correct hysteresis loop with quantities of interest agreeing with other methods in a quantitative manner.

One issue associated to the proposed methods is that nonsymmetric linear systems with variable coefficients are involved 
at each time step. It becomes increasingly difficult to solve such systems using the Generalized 
Minimum Residual Method. However, such linear systems have some unique features, as a consequence of the unique structure
at the continuous level. This shall be used to develop more efficient linear solvers. Meanwhile,
the technique presented here may be applicable to the model for current-driven domain wall dynamics~\cite{ChenGarciaCerveraYang:2015}
and the Schr\"{o}dinger-Landau-Lifshitz system \cite{ChenLiuZhou:2016}, which shall be explored later.

\section*{Acknowledgments}
This work is supported in part by the grants NSFC 21602149 (J.~Chen), AFORS grant FA9550-18-1-0095 (C.~J.~Garc\'{i}a-Cervera), NSF DMS-1418689 (C.~Wang), the Postgraduate Research and Practice Innovation Program of Jiangsu Province via grant KYCX19\_1947 (C.~Xie), and NSFC 11801016 (Z.~Zhou).

\bibliographystyle{model1-num-names}
\bibliography{mybibfile}

\end{document}